\documentclass[journal]{IEEEtran}
\usepackage{ifpdf}
\IEEEoverridecommandlockouts    
\usepackage[latin1]{inputenc}
\usepackage{amsmath}
\usepackage{subfigure}
\usepackage[pdftex]{graphicx}
\usepackage[usenames]{color}
\usepackage{setspace}
\usepackage{comment}
\usepackage{tikz}
\usepackage{bm}
\usepackage{xspace}
\usepackage{soul}
\usepackage{amsopn}
\usepackage{algpseudocode}
\usepackage{algorithm} 
\usepackage{enumerate}
\makeatletter
\let\NAT@parse\undefined
\makeatother
\usepackage{amsmath,amsfonts, amssymb,amsthm,color}
\usepackage[numbers,sort&compress]{natbib}
\usepackage{amsf onts}
\usepackage{color}
\usepackage{setspace}

\usepackage{amssymb}
\newtheorem{theorem}{Theorem}

\newtheorem{assumption}{Assumption}
\newtheorem{proposition}{Proposition}
\newtheorem{corollary}{Corollary}
\newtheorem{lemma}{Lemma}
\newtheorem{remark}{Remark}

\usepackage{breqn}
\usepackage{hyperref} 
\usepackage{epstopdf}
\usepackage{epsfig}
\DeclareGraphicsExtensions{.pdf,.eps,.png,.jpg,.mps}

\newcommand{\p}{\mbf{p}^*(t)}

\newcommand{\R}{\mathbb{R}}

\newcommand{\mbf}[1]{\mathbf{#1}}

\title{\LARGE \bf 
Robustness Analysis for an Online Decentralized \\Descent Power allocation algorithm
}
\author{Chinwendu Enyioha, Sindri Magn\'{u}sson, Kathryn Heal,  Na Li, Carlo Fischione, and Vahid Tarokh
 \thanks{This work was supported by the VR Chromos Project and NSF grant No. 1548204.}
\thanks{C. Enyioha, K. Heal, N. Li, and V. Tarokh are with the School of Engineering and Applied Sciences, Harvard University, Cambridge, MA~USA. (email: cenyioha@seas.harvard.edu; kathrynheal@g.harvard.edu; nali@seas.harvard.edu; vahid@seas.harvar.edu) }%
\thanks{S. Magn\'{u}sson and C.~Fischione are with Electrical Engineering School, Access Linnaeus Center, KTH Royal Institute of Technology, Stockholm, Sweden. (e-mail:{ sindrim@kth.se; carlofi@kth.se})}%
}
\begin{document}
\maketitle


\begin{abstract}
As independent service providers shift from conventional energy to renewable energy sources, the power distribution system will likely experience increasingly significant fluctuation in supply, given the uncertain and intermittent nature of renewable sources like wind and solar energy. These fluctuations in power generation, coupled with time-varying consumer demands of electricity and the massive scale of power distribution networks present the need to not only design real-time decentralized power allocation algorithms, but also characterize how effective they are given fast-changing consumer demands and power generation capacities. In this paper, we present an Online Decentralized Dual Descent (OD3) power allocation algorithm and determine (in the worst case) how much of observed social welfare and price volatility can be explained by fluctuations in generation capacity and consumer demand. Convergence properties and performance guarantees of the OD3 algorithm are analyzed by characterizing the difference between the online decision and the optimal decision. The theoretical results in the paper are validated and illustrated by numerical experiments using real data.

 \end{abstract}

  \section{Introduction}
As the quest to integrate more renewable energy sources into the power distribution grid continues, it is important to understand how systemic fluctuations both in supply and consumer demand affect real-time power allocation policies and the social welfare of the distribution systems.
Some of the challenges of integrating renewable energy sources into the grid have been well-documented  \cite{klessmann2008pros,schleicher2012renewables}. Though research on wind forecasting to reduce uncertainty in day ahead schedules for wind power generation exists, variability in the wind resource continues to pose challenges for the integration of wind power in forward electricity markets \cite{anderson2008reducing}. A similar conclusion can be made of solar energy.
The market value of variable renewable energy has also been analyzed; for instance, in \cite{hirth2013market}, where the authors characterized how the market value of renewable energy sources varies with grid integration. Their study found that the value of wind power fell from $110\%$ of the average power price to about $50-80\%$ as wind penetration increases from zero to $30\%$ of total electricity consumption. This study and results on price fluctuation highlights the difficulty in integrating large-scale renewable energy into the distribution grid. Similar studies in \cite{cox2009impact,mount2010hidden} noted the significant impact of variability of renewable energy sources on the operations of electricity markets. While long-term plans to efficiently integrate renewable energy into future energy infrastructure remains on the horizon \cite{boie2014efficient}, large-scale energy storage mechanisms to soften the variations in supply of energy from renewable sources and the challenges faced on that front have been proposed \cite{castillo2014grid}.

The uncertainties in generation capacity from renewable energy sources and their effects on the electricity markets discussed in \cite{hirth2013market,cox2009impact,mount2010hidden},  present the need for design of fast, efficient and scalable decentralized  \textit{real-time} power allocation algorithms that are robust to inherent systemic fluctuations resulting from  consumers' constantly changing power needs and unstable power generation from renewable energy sources. The current mechanism of a two-way communication in coordinating  decentralized power allocation for such large systems between suppliers and users where the suppliers and users iteratively communicate and carry out computations until they reach an agreement to trigger an event is not only expensive in terms of communication overhead, but also costly in terms of the time it takes to coordinate. This is not preferred, especially given the need to track fast system fluctuations such as variability in generation \cite{lin2012online,mei2011robust} and consumer demands. We propose an alternate approach to the aforementioned coordination scheme in which the decentralized power allocation  algorithm is implemented in real-time: At each time-step, the suppliers update their coordination signal (which is usually able to be interpreted as price signals), and users determine optimal allocations based on the price signal and individual needs. 

With this approach, a number of interesting questions arise, including performance guarantees of the decentralized algorithm in real-time implementation. Since the system (comprising suppliers and users') respectively have constantly changing capacities and power needs/utility, the coordination signal (price) and the power allocation will fluctuate and it is unclear how optimal the online decisions will be. In this paper, we focus on investigating and characterizing performance guarantees of an Online Decentralized Dual Descent (OD3) power allocation algorithm. In particular, we assume the power supply and consumers' utilities fluctuate at the same time scale as iterations of our algorithm; and derive a bound on social welfare of users in the system, based on the OD3 algorithm.

  To prevent drastic changes in both user utility functions and supplier capacities between successive time-steps, we make appropriate smoothness assumptions on the magnitude of fluctuations observed in changes in user utility functions and available supply. The decentralized, iterative algorithm presented in this paper avoids the typical two-way message-passing coordination technique in the literature \cite{kelly1998rate,low1999optimization}, by exploiting the underlying physical characteristics of the problem. Specifically, the suppliers broadcast a price signal, which the users use in computing their optimal consumption at the next time step. Rather than wait for a message on consumption information, the suppliers measure the current power usage to determine the price for the next cycle. A major challenge is that during the process, the system environment may change in that -- power capacity and users' utility changes. In \cite{magnusson2015distributed}, we investigated a special case of the OD3 algorithm -- the static case, whereas in this paper we consider the more challenging \textit{real-time} case. 
  
The rest of the paper is organized as follows: Following notation and definitions, we introduce the system model as well as underlying assumptions, summarize the OD3 power allocation algorithm and state the main result in Section \ref{sec:model}. Sections \ref{sec:market-regulation} and \ref{sec:convergence-analysis} respectively present a volatility analysis of the coordinating (price) signal and power allocation variable and convergence analysis of our algorithm. We follow in Section \ref{sec:numerical} with numerical illustration of our algorithm and make our conclusions in Section \ref{sec:conclude}. \\

    \noindent \textit{Notation:} Vectors and matrices are represented by boldface lower and upper case letters, respectively. We denote the set of real numbers by $\mathbb{R}$, a vector or matrix transpose as $(\cdot)^T$, and the L$2$-norm of a vector by $||\cdot||$. The gradient of a function $f(\cdot)$ is denoted $\nabla f(\cdot)$, and $\langle\ \cdot ,\cdot \rangle$ denotes the inner product of two vectors. We denote the vector of ones as $\mathbf{1}$. \\
  
\noindent \textit{Definitions}

\noindent We say that a function $U:\R^R\rightarrow \R$ is $\sigma$-\textit{strongly convex} if for all $\mbf{q}_1,\mbf{q}_2\in \R^R$ we have 
  \begin{align}
     \langle \nabla U(\mbf{q}_1) - \nabla U(\mbf{q}_2),   \mbf{q}_1 - \mbf{q}_2\rangle \geq \sigma \|\mbf{q}_1 - \mbf{q}_2\|^2, \ \sigma >0. \nonumber
   \end{align}
   
\noindent A function $U(\mbf{q})$ is strongly concave if $-U(\mbf{q})$ is \textit{strongly convex}, i.e., if for all $\mbf{q}_1,\mbf{q}_2\in \R^R$, we have 
     \begin{align} \label{eq:strongly-concave}
      -\langle \nabla U(\mbf{q}_1) - \nabla U(\mbf{q}_2),   \mbf{q}_1 - \mbf{q}_2\rangle \geq \sigma \|\mbf{q}_1 - \mbf{q}_2\|^2. \nonumber
   \end{align}

\noindent A function $U(\cdot)$ is said to be \textit{monotone decreasing}, if for all $\mbf{q}_1 \leq \mbf{q}_2$, it holds that $U(\mbf{q}_1) \geq U(\mbf{q}_2) $.

\section{Model and Algorithm}
\label{sec:model}
\subsection{System Model}
We consider a power distribution system comprising $N$ users and $R$ power producers. The objective at time $t$ is to solve a decentralized dynamic power allocation problem to maximize the aggregate social welfare of the system.  We  define the aggregate social welfare as the utilities users gain from consuming power $\mbf{q}_i(t)$ at time $t$; that is, $\sum_{i=1}^N U_i^t(\mbf{q}_i(t))$. We abstract away specific power flow constraints and assume that each user and supplier respectively have a dynamic utility function and capacity. Let the allocation of the $N$ users be $\mbf{q}_1(t),\cdots,\mbf{q}_N(t)$, where $\mbf{q}_i(t)\in \R^R$. We assume each user can choose from the $R$ available power suppliers. The $j$th entry of the $i$th vector, $\mbf{q}_i^j(t)$ represents the power allocation to the $i$th user from the $j$th supplier at time $t$. Furthermore, let the capacity of the power suppliers at time $t$ be represented by the vector $Q(t) \in \mathbb{R}^R$, where $Q_j$, the $j$th entry of $Q$, represents the power supply capacity at the $j$th supplier. At time $t$, the objective of the system operator is to maximize some strongly concave utility function of users' power allocation, $U_i^t(\mbf{q}_i(t))$. 
We formulate the optimal power allocation problem as the following optimization program for each time $t$:
\begin{equation}
\begin{aligned}\label{eq:system-problem1}
& \underset{\mbf{q}_1(t),\hdots,\mbf{q}_N(t)}{\text{maximize}}
& & \sum_{i=1}^N U_i^t(\mbf{q}_i(t)) \\
& \text{subject to}
 & & \sum_{i=1}^N \mbf{q}_i(t) = Q(t). \\
\end{aligned}
\end{equation}
The static case of \eqref{eq:system-problem1} where the utility functions and supplier capacities are fixed has been studied in the literature; for example, \cite{magnusson2015distributed}, where the system comprised a single power producer and $N$ users and conditions that guarantee feasibility of the power allocation problem at each step of the iterative solution were derived. 

In Problem \eqref{eq:system-problem1}, the decision variable at each user $\mbf{q}_i$ is unrestricted, because we assume that users are able to sell power to the distribution system.  The local utility function $U^t_i(\cdot)$ of users and power generation capacity $Q(t)$ on the supply end are usually time-varying, which poses a challenge to solving Problem \eqref{eq:system-problem1}, because optimization algorithms are usually iterative. Another challenge is that the utility functions $U_i^t(\cdot)$ are known locally by user $i$; hence, solving Problem \eqref{eq:system-problem1} requires a decentralized algorithm. 
In the proposed algorithm, for clarity in presentation, we will use the following matrix $\mbf{Q} = [\mbf{q}_1, \hdots, \mbf{q}_N ] \in \mathbb{R}^{R\times N}$ to compute the aggregate power allocation by each supplier and compute the coordinating (or price) signal at the next time step.
\begin{assumption} (Strong Concavity): We assume the users' utility functions, $U_i^t(\mbf{q})$, are strongly concave in the variable $\mbf{q}$ with parameter $\sigma_i^t$. 
\label{assum:strong-concave}
\end{assumption}

\begin{assumption} (Lipschitz Gradients): The gradients of the utility function of each user $i$ is Lipschitz continuous at each time-step. In other words, for all vector pairs $\mbf{q}_1$ and $\mbf{q}_2$, $\|\nabla U_i^t(\mbf{q}) - \nabla U_i^t(\mbf{q}_2) \| \leq  L_i^t \|\mbf{q}_1 - \mbf{q}_2 \|$, where $L_i^t < \infty$ is the Lipschitz constant.
\label{assum:Lip-grad}
\end{assumption}
Assumptions \ref{assum:strong-concave} and \ref{assum:Lip-grad} above imply that the (local) objective function in Problem \eqref{eq:system-problem1} is strongly concave with Lipschitz gradients. We will assume that $\sigma$ and $L$ are respectively the global concavity and Lipschitz parameter for $U_i^t(\cdot)$ and $\nabla U_i^t(\cdot)$; that is, 
\begin{equation}
\sigma = \min_{i, t} \{\sigma_i^t\} \quad \text{and} \quad L = \max_{i,t} \{L_i^t\}. \label{eq:global}
\end{equation}
  Problem \eqref{eq:system-problem1} can easily be generalized to account for the changing cost of power production by the suppliers. In this paper, we assume the utility functions of users in the system and capacity of power suppliers are time-changing; and to prevent drastic variations in the users' utility functions and suppliers' capacities, we make the following smoothness assumptions:
\begin{assumption}\label{assum:capacity-bound}
Between successive time-steps the changes in suppliers' capacities is bounded by $\gamma$; that is,
\begin{equation} \label{regulation:gamma}
\|Q(t) - Q(t+1) \| \leq \gamma,  \ \forall \ t.
\end{equation}
\end{assumption}

\begin{assumption}\label{assum:utility-bound}
We assume that each user $i$ has an upper bound on how much its utility function changes between consecutive time steps; that is,
\begin{equation} \label{regulation:alpha}
\|\nabla U_i^{t+1}(\mbf{q}_i) - \nabla U_i^t(\mbf{q}_i)\| \leq \alpha, \quad \forall \ i, \mbf{q}_i \ \text{and} \ t.
\end{equation}
\end{assumption} 
Assumption \ref{assum:capacity-bound} prevents drastic changes in the supply capacities between consecutive time-steps, ensuring a smoothness property in capacity at the supplier over time. And Assumption \ref{assum:utility-bound} ensures that the rate of change in user utility function over time is bounded -- again to avoid drastic changes in power demand between consecutive time-steps for all users. These assumptions, in practice, can be enabled by the use of backup power sources. 
  Our objective is to present a distributed solution to \eqref{eq:system-problem1} using a one-way communication (coordination) protocol that enables us  characterize variation in the social welfare of the system at each time-step as presented in Theorem \ref{thm:social welfare}.
We present a decentralized policy to solve \eqref{eq:system-problem1} using the following operations for coordinating the decentralized allocation:

\noindent \textbf{Operation 1} (One-way Communication): At each time-step $t$, the power suppliers broadcast a message (referred to as the unit cost of power), to the users. Note that the cost of power charged by the $R$ different suppliers are allowed to be different.

\noindent \textbf{Operation 2} (Feedback Information) At each time-step, the Load Servicing Entities measure the difference between their total power capacity and their supply; that is, $\mbf{Q}(t)\mbf{1} - Q(t)$.
We investigate the performance of the well-known dual descent algorithm for Problem \eqref{eq:system-problem1}. Decentralized algorithms that solve \eqref{eq:system-problem1} using Operations 1 and 2 can be achieved via duality theory. 
Let $\mbf{p}(t)\in \mathbb{R}^R$ be the dual variable representing the price charged by the $R $ different suppliers at time $t$. Then, the dual problem of \eqref{eq:system-problem1} is
\begin{equation}
\begin{aligned}\label{eq:dual}
& \underset{\mbf{p}}{\text{minimize}}
& & D^t(\mbf{p}(t))
\end{aligned}
\end{equation}
where $D^t(\cdot)$ is the dual function and $D^t(\cdot)$ is:
\begin{align}
D^t(\mbf{p}(t)) =& \max_{\mbf{q}_i} \quad\mathcal{L}(\mbf{q},\mbf{p}), \label{eq:compact-dual}
\end{align}
where $$\mathcal{L}(\mbf{q},\mbf{p}) =  \left[ \sum_{i=1}^N U_i^t(\mbf{q}_i(\mbf{p}(t))) {-} \mbf{p}(t)^T \left(   \sum_{i=1}^N \mbf{q}_i(\mbf{p}(t)) {-} Q(t)  \right) \right],$$ and $\mbf{q}_i(\mbf{p}(t))$ is the power demand of user $i$ based on price $\mbf{p}(t)$ at time $t$.
The respective local problem for each user $i$ is to solve:
\begin{align}
\mbf{q}_i(\mbf{p}(t)) & = \arg \max_{\mbf{q}_i} \ \left[ U_i^t(\mbf{q}_i) - \mbf{p}^T\mbf{q}_i\right]. \label{eq:local-primal}
\end{align}
The structure of the problem enables us to claim the following result:

\begin{lemma}\label{lem:strong-duality}
(Strong Duality): Consider Problem \eqref{eq:system-problem1} at time $t$, and suppose Assumptions \ref{assum:strong-concave} and \ref{assum:Lip-grad} hold, and let $\mbf{p}^*(t)$ be the optimal solution to \eqref{eq:dual}, then $\mbf{q}(\mbf{p}^*(t)) = \{\mbf{q}_i(\mbf{p}^*(t))\}_{i=1}^N$ (cf. \eqref{eq:local-primal}) is the optimal solution to \eqref{eq:system-problem1}.
\end{lemma}

  \begin{proof}
     Convexity of the problem coupled with the constraints $\sum_{i=1}^N \mbf{q}_i(t) = Q(t)$, and $\mbf{q}_i \in \mathbb{R}^R$ ensures that~\eqref{eq:system-problem1} satisfies Slater's condition, yielding a zero duality gap ~\cite[Chapter 5]{convex_boyd}.
  \end{proof}

\begin{proposition}\label{prop:strongly-convex-dual}
Consider Problem \eqref{eq:system-problem1} and suppose Assumptions \ref{assum:strong-concave} and \ref{assum:Lip-grad} hold, then the dual function \eqref{eq:compact-dual} is strongly convex in $\mbf{p}$ with parameter $N/ L$, and its gradient is Lipschitz continuous with parameter $N\sigma$.
\end{proposition}
\begin{proof}
Details are presented in Appendix \ref{app:proof-prop:strongly-convex-dual}.
\end{proof}

We solve Problem \eqref{eq:system-problem1} in a decentralized manner via  a dual descent algorithm with $\eta > 0$ as step size. Based on aggregate consumption at the previous time-step, suppliers determine and broadcast the coordinating signal (price). Users use that information to  compute their optimal allocation using their current utility functions. By exploiting the underlying physical structure of the network, the supplier measures the total consumption and computes the next price. We summarize the Online Dynamic Decentralized Dual Descent (OD3) Algorithm in Algorithm \ref{alg:algorithm}. 
\begin{algorithm}
\caption{An Online Decentralized Dual Descent (OD3) Algorithm for optimal power allocation}
\label{alg:algorithm}
\begin{algorithmic}[1]
\Statex \textbf{Initialization}: Suppliers set initial price $\mbf{p}(0)$; 
\Statex and let $\eta \in ]0, \bar{\eta}]$ be given.
\For {$t=0, \hdots $}
\State Suppliers broadcast $\mbf{p}(t)$
\For {Users $i=1,\hdots, N$}
\State User $i$ receives $\mbf{p}(t)$ and solves \eqref{eq:local-primal}
\EndFor
\State Suppliers measure $\mbf{Q}(t)\mbf{1} - Q(t)$ 

and compute next price
\State $\mbf{p}(t+1) = \mbf{p}(t) - \eta (\mbf{Q}(t)\mbf{1} - Q(t))  $
\EndFor
\end{algorithmic}
\end{algorithm}
Given Algorithm \ref{alg:algorithm} to solve Problem \eqref{eq:system-problem1}, of interest is to understand and characterize how the real-time and optimal system decisions and social welfare of the system change with time, given the time-varying capacities and utility functions. This brings us to our main result.
\begin{theorem}\label{thm:social welfare} (Main Result)
Suppose Algorithm \ref{alg:algorithm} with step-size $0 < \eta \leq 2L/(N(1+L\sigma))$ is used to solve Problem \ref{eq:system-problem1}, suppose Assumptions \ref{assum:capacity-bound}, and \ref{assum:utility-bound} hold and that each $U_i^t(\cdot)$ is Lipschitz with parameter $L_i^{'t}$. Suppose $L' = \max_{i,t} \{L_i^{'t}\}$. Let $\mbf{q}_i(t)$ and $\mbf{q}_i^*(t)$ respectively be the power allocation obtained from Algorithm \ref{alg:algorithm} and the optimal power allocation at time $t$. The difference between the aggregate online social welfare (given fluctuations in the system) and aggregate optimal social welfare is bounded by
\begin{equation}
 \left\Vert \sum_{i=1}^N U_i^t(\mbf{q}_i(t)) - \sum_{i=1}^N U_i^t(\mbf{q}_i^*(t)) \right\Vert \leq W,
\end{equation}
where $$W  = NL' \left[\frac{c^t}{\sigma}||\mbf{p}(0)-\mbf{p}^*(0)||+\frac{L^2}{\sigma^2}\left(\frac{\gamma}{N}+\frac{\alpha}{\sigma}\right)\right],$$
$$ \text{and} \quad c=\left(1{-}\frac{2\eta\sigma N}{(1+\sigma L)}\right)^{1/2}. $$
 Furthermore, $\mbf{p}(0)$ and $\mbf{p}^*(0)$ are respectively the initial online and optimal prices, $\gamma$ and $\alpha$ are respectively the bounds on variations of the suppliers' capacity and users' utility functions.
  \end{theorem}
\begin{remark}
Note that for Theorem \ref{thm:social welfare}, we assume that the utility functions themselves (and not their gradients as in Assumption \ref{assum:Lip-grad}) are Lipschitz.
\end{remark}

\noindent The rest of the paper develops interesting theory that can, in particular, be used to prove Theorem \ref{thm:social welfare}, which is presented in Section \ref{proof:social-welfare}.

\section{Volatility Analysis}\label{sec:market-regulation}
 In this section we derive bounds on changes in the system operator's optimal decisions $\mbf{p}^*(t)$ between successive time-steps based on Assumptions \ref{assum:capacity-bound} and \ref{assum:utility-bound}.

\begin{theorem}\label{lem:price-bound}(Volatility of the Optimal Price):
Consider Problem~\eqref{eq:system-problem1}
with dual problem \eqref{eq:dual}.
  Suppose Assumptions~\ref{assum:Lip-grad} -- \ref{assum:utility-bound} hold; then 
\begin{equation}
  \|\mbf{p}^*(t) - \mbf{p}^*(t+1) \| \leq \frac{L^2}{\sigma}\left(\frac{\gamma}{ N} + \frac{\alpha}{\sigma}\right).
\end{equation}
\end{theorem}
 \begin{proof}
  Given that the dual problem~\eqref{eq:dual} is unconstrained, convex and differentiable, for all $t$ 
  \begin{equation} \label{inProof-PR:optimality-cond}
    \nabla D^t(\mbf{p}^{\star}(t)) = \sum_{i=1}^N [\nabla U_i^t]^{-1}(\mbf{p}^*(t))  - Q(t) = 0. \nonumber
   \end{equation}
   In particular, let $\Gamma_t(\mbf{p}) =\sum_{i=1}^n [\nabla U_i^t]^{-1}(\mbf{p}) $; then   $\mbf{p}^*(t) = \Gamma_t^{-1}(Q(t))$. 
    And by Lemma~\ref{bijLemma} (see Appendix), the inverse function of the gradient exists.
   Using the triangle inequality, it follows that
   \begin{align*}
       \|\mbf{p}^*(t) - \mbf{p}^*(t+1) \| = & || \Gamma_t^{-1}(Q(t)) -\Gamma_{t+1}^{-1}(Q(t{+}1)) ||  \\
              \leq   & || \Gamma_t^{-1}(Q(t)){-} \Gamma_t^{-1}(Q(t{+}1))|| \\
               & + ||  \Gamma_t^{-1}(Q(t{+}1))  {-}\Gamma_{t+1}^{-1}(Q(t{+}1)) || \\
               \leq & \frac{L^2 \gamma}{\sigma N} + \frac{\alpha L^2}{\sigma^2} \\
               =& \frac{L^2}{\sigma}\left(\frac{\gamma}{N} + \frac{\alpha}{\sigma}  \right). \qedhere
   \end{align*}
 \end{proof}
The first term in the second inequality above is obtained from the fact that $\Gamma^{-1}$ is $L^2/(\sigma N)$-Lipschitz continuous from Lemma \ref{bijLemma}-d (in Appendix \ref{app:additional-lemmas}). And the second term comes from \eqref{inTVlemma:1} of Lemma \ref{Lemma:varying} ( also in Appendix \ref{app:additional-lemmas}).

Theorem \ref{lem:price-bound} above implies that price fluctuation is higher when the changes in supply capacity $\gamma$ and user utility functions $\alpha$ are high. Furthermore, the presence of more users in the system, captured by $N$ in the first term of the bound dampens any significant price changes observed between consecutive time-steps caused by the time-varying supplier capacities. In addition, the more concave the form of the utility function of users (for $\sigma >1$), the lower the volatility in the price (coordinating) signal between consecutive time-steps. 

Given that the power allocation $\mbf{q}_i(\mbf{p}(t))$ at any time $t$ depends on the price at that time, a consequence of Theorem \ref{lem:price-bound} are bounds on the optimal primal variables, which we present next.

\begin{corollary}\label{prop:primal-bound}
Consider problem \eqref{eq:system-problem1}, and given Assumptions \ref{assum:strong-concave} and \ref{assum:Lip-grad}, the optimal power allocation between consecutive iterates, satisfies
$$ \|\mbf{q}_i^*(t+1) - \mbf{q}_i^*(t) \| \leq \frac{ L^2}{\sigma^2} \left(\frac{ \gamma}{ N} + \frac{\alpha}{\sigma}\right)  +\frac{\alpha}{\sigma} .$$
\end{corollary}
\begin{proof}
See Appendix \ref{app:prop:primal-bound}
\end{proof}
Similar to the bound in Theorem \ref{lem:price-bound}, we find in Corollary \ref{prop:primal-bound} that the magnitude of changes in power allocation between consecutive time steps is proportional to the changes experienced in supply capacities and users' utility functions. Furthermore, as highlighted earlier, when no fluctuations are recorded in users' utility functions, the size of the user base mitigates the changes in power demand between successive time steps.
\section{Convergence Analysis of Algorithm \ref{alg:algorithm}}
\label{sec:convergence-analysis}
In this section, we analyze convergence of Algorithm \ref{alg:algorithm}. In particular, derive an upper bound on the difference between the price computed by Algorithm \ref{alg:algorithm} and the optimal price at each iteration. In \cite{magnusson2015distributed}, we showed how to appropriately select parameters -- such as the step size and initial price, to achieve a linear convergence rate in the decentralized solution of the static case of Problem \eqref{eq:system-problem1}. 

\begin{lemma}\label{thm:convergence}
 For every step-size $0 < \eta \leq 2L/(N(1+L\sigma))$, and $\mbf{p}\in \R^R$, we have
   \begin{align} \label{eq:step-size}
        || \mbf{p}^*(t) - (\mbf{p} - \eta \nabla D_t(\mbf{p}) )||^2  \leq c^2|| \mbf{p}^*(t) - \mbf{p} ||^2
   \end{align}
   where $c$ is as defined in Theorem \ref{thm:social welfare}.
\end{lemma}
\begin{proof}
Choose any $\mbf{p}\in\mathbb{R}^R.$ Let $r_t:=\|\p-\mbf{p}\|.$ Then
 \begin{equation*}\begin{split}
	||\p&-(\mbf{p}-\eta\nabla D_t(\mbf{p}))||^2\\
        &= r_t^2+2\eta\langle \p-\mbf{p},\nabla D_t(\mbf{p})\rangle + \eta^2||\nabla D_t(\mbf{p})||^2\\
        &\leq \left( 1{-}\frac{2\eta N \sigma}{1{+}L\sigma} \right)r_t^2 {+} \eta\left(\eta{-} \frac{2L}{N(1{+}L\sigma)}   \right)\|\nabla D_t(\mbf{p})\|^2,
        \end{split}
        \end{equation*}
        where we have used \cite[Theorem 2.1.12]{nesterov2004introductory}, the fact that $D(\cdot)$ has $N\sigma$-Lipschitz continuous gradients, is strongly convex with parameter $N/L$; the fact that an $L$-Lipschitz convex function $D(\cdot)$ satisfies 
        $$ 0\leq D(\p) - D(\mbf{p}) - \langle \nabla D(\mbf{p}), \p - \mbf{p} \rangle  \leq \frac{L}{2} \|\mbf{p} - \p \|,   $$ and that at the optimal $\p$, $\nabla D_t(\p) = 0$.
        If the step-size $\eta$ is chosen such that $ 0<\eta < 2L/N(1+L\sigma)$, 
        one always obtains $\eta\left(\frac{2 L}{N(1+\sigma L)}-\eta\right)||\nabla D_t(\mbf{p})||^2\geq 0$; thence,
        \begin{equation}\begin{split}
	||\p{-}(\mbf{p}{-}\eta\nabla D_t(\mbf{p}))||^2\leq\left(1{-}\frac{2\eta\sigma N}{(1+\sigma L)}\right)^2||\p{-}\mbf{p}||^2, \nonumber
        \end{split}\end{equation}
        and it follows that, for each iterate $\mbf{p}$ obtained from the OD3 Algorithm \ref{alg:algorithm}, the bound in \eqref{eq:step-size} holds.
\end{proof}
The choice of the appropriate step-size $\eta$ results in the parameter $c=1-2\eta\sigma N/((1+\sigma L))<1$.
 Having  established convergence of the OD3 Algorithm \ref{alg:algorithm}, we present a result on real-time tracking the optimal primal and dual variables, given the dynamic supplier capacities and user utility functions.

\begin{theorem}\label{thm:main-result}(Tracking the Optimal Dual Variable): 
Consider the system Problem \eqref{eq:system-problem1} and its Dual problem \eqref{eq:dual}. Irrespective of the initial price $\mbf{p}(0)$, the distance between the optimal price $\mbf{p}^*(t+1)$ and the price iterate of Algorithm \ref{alg:algorithm} at each $t$ is 
\begin{equation}\label{eq:price-tracking}
\|\mbf{p}(t+1) - \mbf{p}^*(t+1) \| \leq  \frac{b}{1-c} + c^t \left(\|\mbf{p}(0) - \mbf{p}^*(0)\| - \frac{b}{1-c}  \right),
\end{equation}
where from Theorem \ref{lem:price-bound} and Lemma \ref{thm:convergence} respectively, 
\begin{equation}\label{eq:b-and-c}
b = L^2\left(\frac{\gamma}{\sigma N} + \frac{\alpha}{\sigma^2} \right) \quad \text{and} \quad c = \left(1{-}\frac{2\eta\sigma N}{(1+\sigma L)}\right)^{\frac{1}{2}}.
\end{equation}
\end{theorem}
The bound in Eq. \eqref{eq:price-tracking} above has two components -- a transient term (similar to the dual gradient algorithm), which converges to $0$ as $t$ goes to $\infty$. Furthermore, the second component of the bound is due to system fluctuations and never becomes $0$ over time.
\begin{proof}
Using the triangle inequality, we can express the LHS of \eqref{eq:price-tracking} as 
$$\|\mbf{p}(t{+}1){-} \mbf{p}^*(t{+}1)  \| \leq \|\mbf{p}(t{+}1){-}\mbf{p}^*(t) \|{+} \|\mbf{p}^*(t){-}\mbf{p}^*(t{+}1) \|. $$
Observe that bounds for the first and second summands of the RHS above have been respectively derived in Lemma \ref{thm:convergence} and Theorem \ref{lem:price-bound}. Since $\|\mbf{p}^*(t){-}\mbf{p}^*(t{+}1) \| \leq b$ (from Theorem \ref{lem:price-bound}), one obtains $\|\mbf{p}(t{+}1){-} \mbf{p}^*(t{+}1)  \| \leq \|\mbf{p}(t{+}1){-}\mbf{p}^*(t) \| + b $, which further simplifies into $\|\mbf{p}(t{+}1){-} \mbf{p}^*(t{+}1)  \| \leq c \| \mbf{p}(t) - \mbf{p}^*(t) \| + b $.  Further simplification results in
\begin{align}
\|\mbf{p}(t{+}1){-} \mbf{p}^*(t{+}1)  \| & \leq c^t \|\mbf{p}(0) - \mbf{p}^*(0) \| + \sum_{i=0}^t c^i b \nonumber \\
&= c^t||\mbf{p}(0)-\mbf{p}^*(0)|| + \frac{1-c^t}{1-c} b \nonumber \\
 & = \frac{b}{1- c} + c^t \left(\|\mbf{p}(0) - \mbf{p}^*(0)\| - \frac{b}{1-c}  \right), \label{eq:price-track}
\end{align}
which concludes the proof. 
\end{proof}
In Equation \eqref{eq:price-track} above, the first equality uses the sum of a geometric series, and the last equality separates the bound into a constant and transient term. A similar result is obtained for the primal variable.

\begin{theorem}\label{thm:main-result2}(Tracking the Optimal Primal Variable): 
Consider the system problem \eqref{eq:system-problem1} and its dual problem \eqref{eq:dual}. Regardless of the initial unit of power demanded per user $\mbf{q}_i(0)$, the distance between the optimal unit of power $\mbf{q}_i^*(t+1)$ and the units generated by Algorithm \ref{alg:algorithm} at each $t$ is 
\begin{equation*}\label{eq:power-tracking}
\|\mbf{q}_i(t+1) - \mbf{q}_i^*(t+1) \| \leq  \frac{c^t}{\sigma} \|\mbf{p}(0) - \mbf{p}^*(0)   \| +  \frac{L^2}{\sigma^2}\left( \frac{\gamma}{N} + \frac{\alpha}{\sigma} \right),
\end{equation*}
where the constant $c$ comes from Theorem \ref{thm:main-result}.
\end{theorem}
\begin{proof}
See Appendix \ref{app:thm:main-result2}
\end{proof}
\normalsize
\begin{remark}
As in Theorem \ref{thm:main-result}, the bound on changes in power allocations depends on two terms --  a transient term that converges to $0$ as $t\rightarrow \infty$, and a term that depends on fluctuations in user utility functions and supplier capacity in the distribution system. With the results established thus far, we are ready to prove the main result of this paper earlier stated in Section \ref{sec:model}.
\end{remark}
\subsection{Proof of Main Result (Theorem \ref{thm:social welfare})}
\label{proof:social-welfare}
\begin{proof}
Let each $U_i^t(q)$ be $L'$-Lipschitz continuous. Thus for any $\mbf{q}_1(t),\mbf{q}_2(t)\in\mathbb{R}^{R}$, we have $$|| U_i^t(\mbf{q}_1(t))- U_i^t(\mbf{q}_2(t))||\leq L'||\mbf{q}_1(t)-\mbf{q}_2(t)||.$$ Summing both sides and using the triangle inequality,
\begin{multline*}
\left\vert\left\vert \sum U_i^t(\mbf{q}_1(t))- \sum U_i^t(\mbf{q}_2(t))\right\vert\right\vert \\
\leq \sum\left\vert\left\vert U_i^t(\mbf{q}_1(t))- U_i^t(\mbf{q}_2(t))||\leq  NL'||\mbf{q}_1(t)-\mbf{q}_2(t)\right\vert\right\vert.
\end{multline*}
Suppose $\mbf{q}(t) =\mbf{q}_1(t)$ and $\mbf{q}^*(t) =\mbf{q}_2(t)$; from Theorem \ref{thm:main-result2} we can derive a bound on the Right-Hand-Side (RHS) of the above expression; that is,
\begin{multline*}
\left\vert\left\vert \sum U_i^t(\mbf{q}(t))- \sum U_i^t(\mbf{q}^*(t))\right\vert\right\vert\ 
\leq  NL'||\mbf{q}(t)-\mbf{q}^*(t)|| \\
\leq NL' \left[\frac{c^t}{\sigma}||\mbf{p}(0)-\mbf{p}^*(0)||+\frac{L^2}{\sigma^2}\left(\frac{\gamma}{N}+\frac{\alpha}{\sigma}\right)\right],
\end{multline*}
where 
$$c=\left(1-\frac{2\eta\sigma N}{1+\sigma L}\right)^{1/2}.$$
Since $c<1$ (from Lemma \ref{thm:convergence}), the first term of the bound above goes to zero as $t$ goes to infinity, and the second term is a constant term.
Hence, the deviation of the aggregate online social welfare (computed using the OD3 algorithm), from the aggregate optimal social welfare is bounded.	 
\end{proof}

Next, we analyze how well the iterates of Algorithm \ref{alg:algorithm} satisfy the feasibility conditions of Problem \eqref{eq:system-problem1}; that is, 
$$ \sum_{i=1}^N \mbf{q}_i(t) - Q(t) = 0 $$
The next result characterizes feasibility (constraint violation) of the OD3 Algorithm for Problem \eqref{eq:system-problem1}.
\begin{corollary}
Given Algorithm \ref{alg:algorithm} to solve Problem \eqref{eq:system-problem1}. Let the primal iterates of Algorithm \ref{alg:algorithm} be $\mbf{q}_i(t)$ at each time $t$, the constraint of Problem \eqref{eq:system-problem1} satisfies the following bound
\begin{align*}
\left\Vert\sum_{i=1}^N \mbf{q}_i(t){-}Q(t) \right\Vert \leq \frac{N}{\sigma}\frac{b}{1-c} {+} c^t \left(\|\mbf{p}(0) {-} \mbf{p}^*(0)\|{-} \frac{b}{1-c}  \right).
\end{align*}
\end{corollary}
\begin{proof}
To prove this, note that  Problem \eqref{eq:system-problem1} is a convex program. Further, since its utility functions are strongly concave, the equality constraint is binding, for the set of optimal decision variables $\{\mbf{q}_i^*(t)\}_{i=1}^N$; that is, $\sum_{i=1}^N \mbf{q}_i^*(t) - Q = 0$. Hence, since $\mbf{q}_i(t) = (\nabla U^t_i)^{-1}(\mbf{p}(t)) $ and from Lemma \ref{Lemma:varying}, it follows that
\begin{align}
\left\Vert \sum_{i=1}^N \mbf{q}_i(t) - Q(t) \right\Vert & = \left\Vert\sum_{i=1}^N \mbf{q}_i(t) - \sum_{i=1}^N \mbf{q}_i^*(t) \right\Vert \\
 &= \|\Gamma_t(\mbf{p}(t)) - \Gamma_t(\mbf{p}^*(t))  \| \nonumber \\
& \leq \frac{N}{\sigma} \| \mbf{p}(t) - \mbf{p}^*(t)  \|, \label{eq:constraint-violation}
\end{align}
since $\Gamma_t(\cdot)$ is $N/\sigma$-Lipschitz continous from Lemma \ref{bijLemma}-c and the statement of the result follows from Theorem \ref{thm:main-result}.
\end{proof}
Therefore, at each iteration of the algorithm, the constraint violations are bounded by \eqref{eq:constraint-violation}. The bound is intuitive since the power allocation depends on the price at that time instance and more users in the system results in increased violation of the allocation-capacity equality constraint.

\begin{figure*}[!tbp]
  \centering
  \begin{minipage}[b]{0.332\textwidth}
    \includegraphics[width=\textwidth]{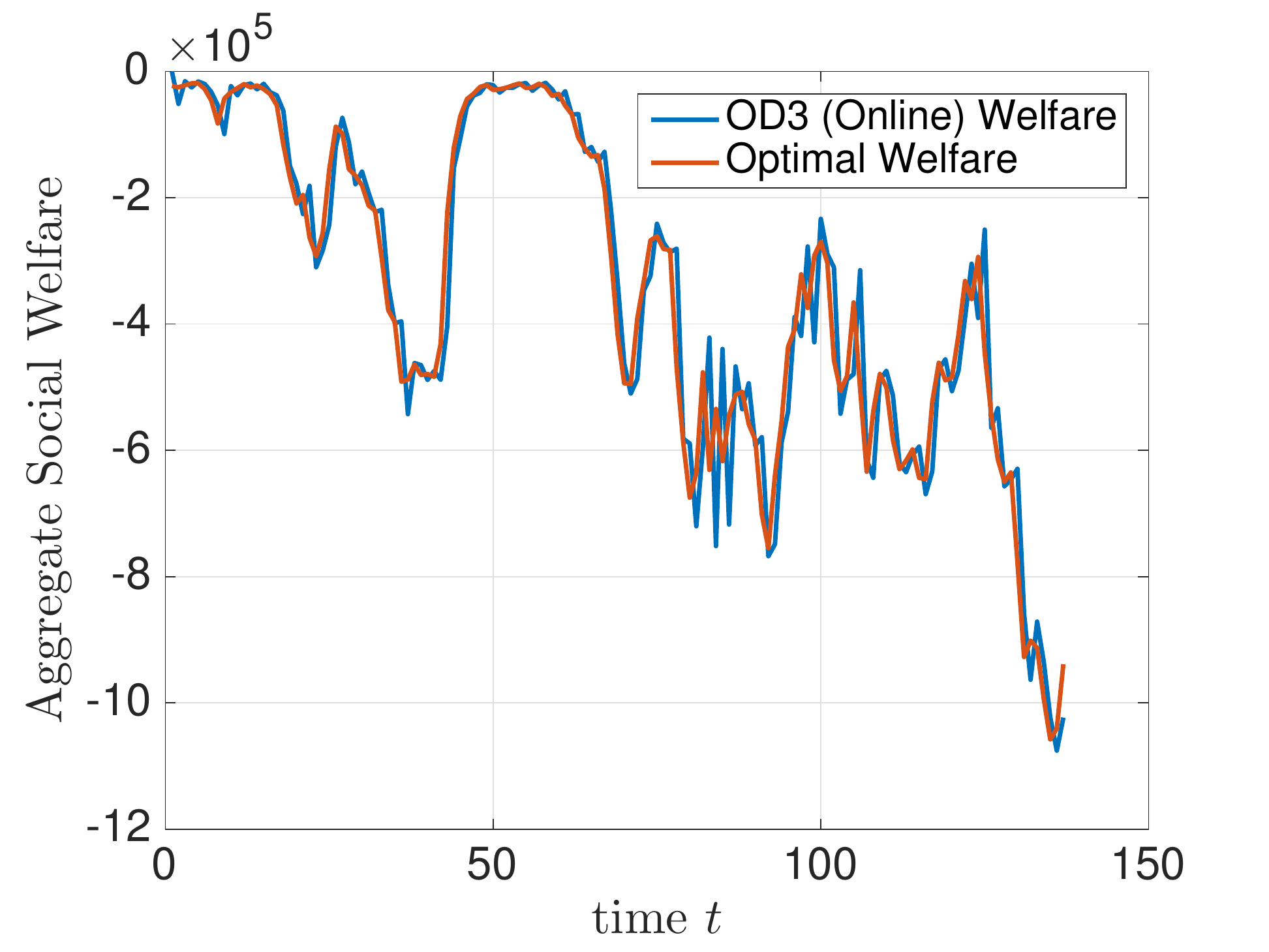}
    \caption{A plot showing the social welfare of the system computed by the OD3 algorithm in real-time and the optimal social welfare over time. }
    \label{fig:social}
  \end{minipage}
  \hfill
  \begin{minipage}[b]{0.32\textwidth}
    \includegraphics[width=\textwidth]{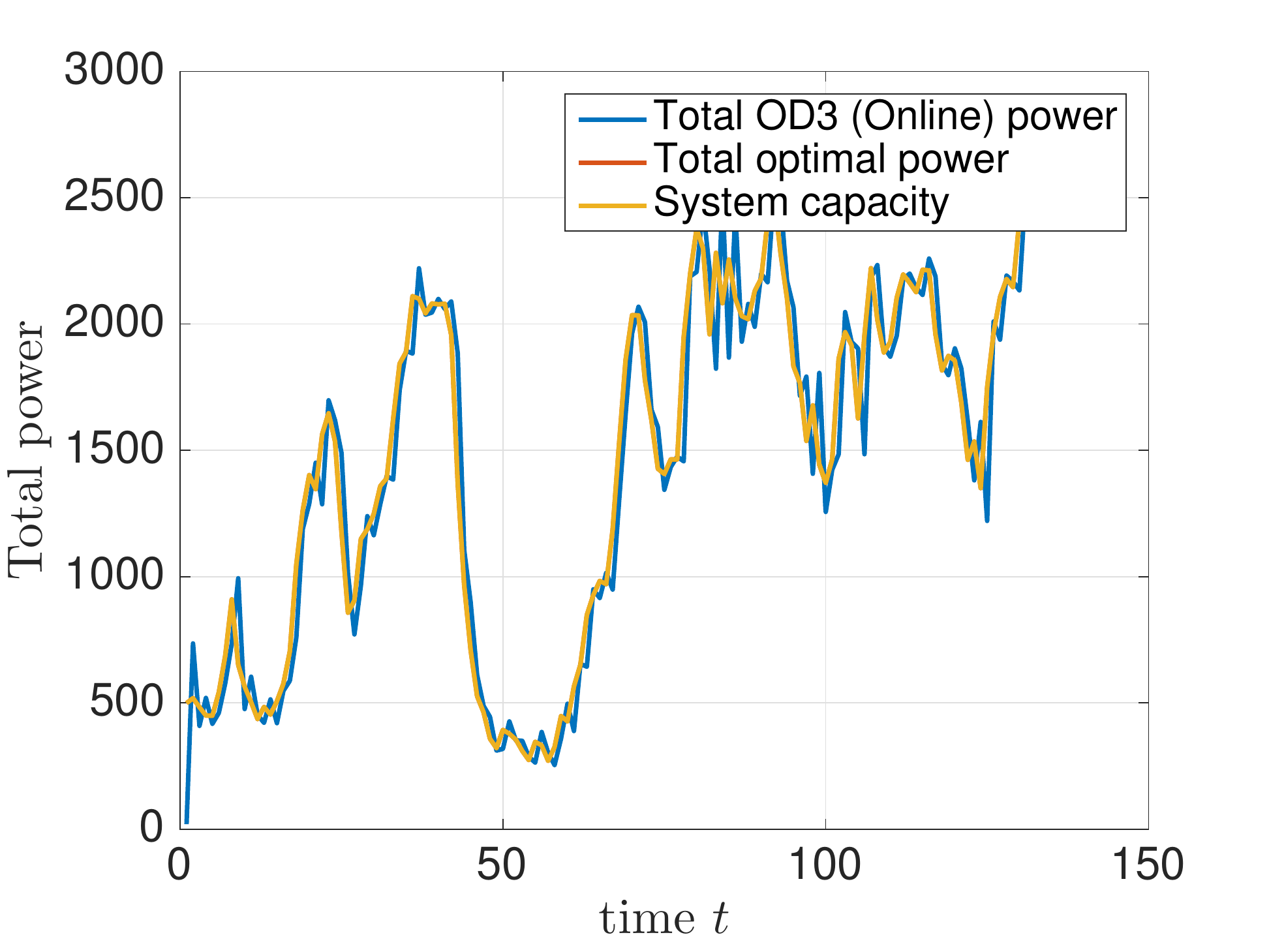}
    \caption{This figure shows the the total power allocated to users in real-time in comparison to the optimal total power and the total power supply available over time.}
    \label{fig:power}
  \end{minipage}
  \hfill
  \begin{minipage}[b]{0.32\textwidth}
    \includegraphics[width=\textwidth]{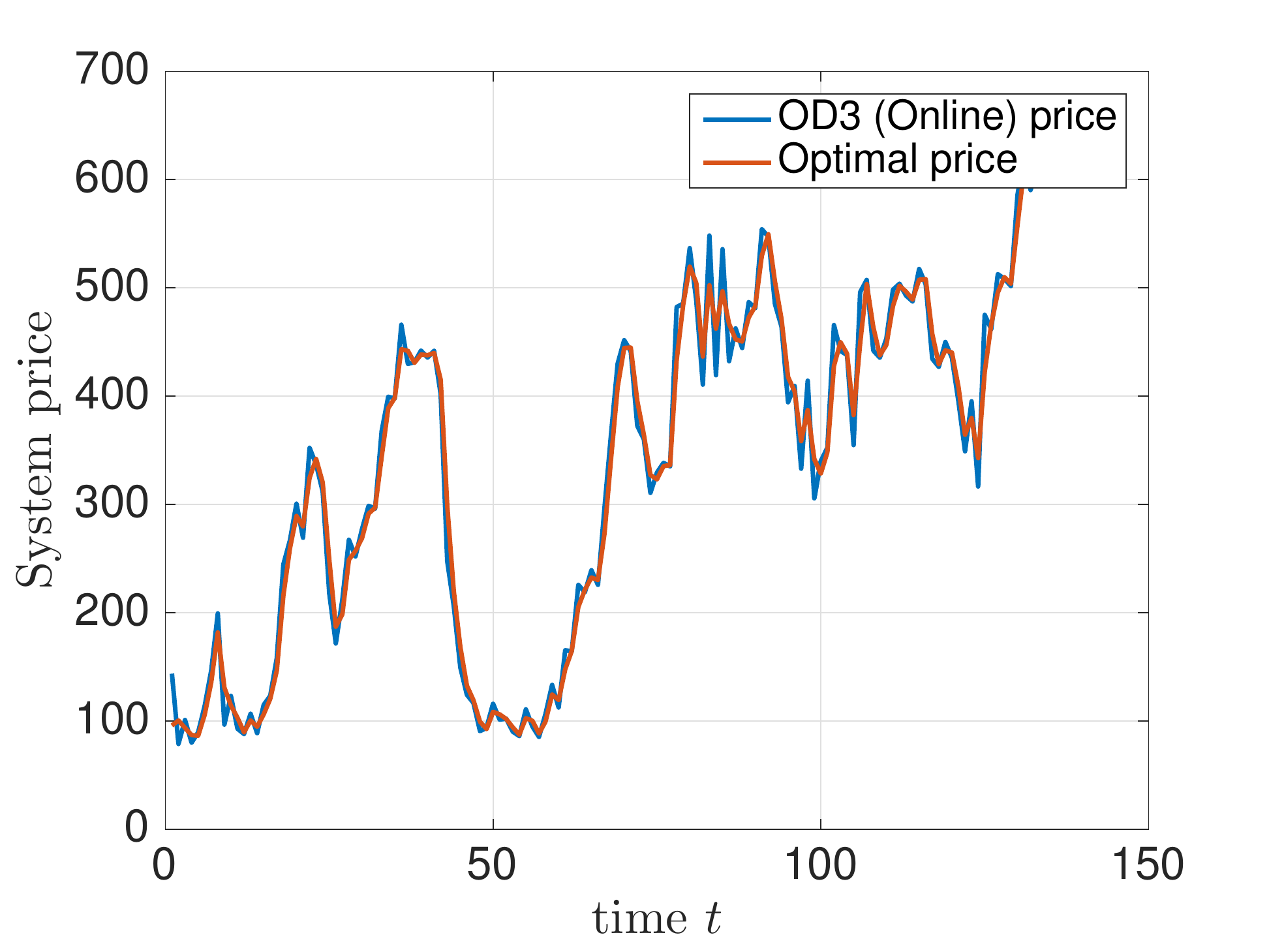}
    \caption{This figure illustrates the prices at each time-step generated by the real-time OD3 algorithm and the optimal price.}
    \label{fig:price-error}
  \end{minipage}
\end{figure*}


\section{Numerical Illustration}\label{sec:numerical}

We illustrate the performance of the OD3 power allocation algorithm on Problem \eqref{eq:system-problem1} comprising $N=10$ users and one supplier. For each user $i=1,\hdots, 10$, we set $U_i^t(\mbf{q}_i(t)) = -(\mbf{q}_i(t) - \mbf{s}_i^t)^2$. The variation in the utility functions is captured by the time-varying term $\mbf{s}_i^t$; and this choice of utility function $U_i^t(\cdot)$ is strongly concave. The \real{real data} on power supply was obtained from \cite{ieso}. The power generation data was from biofuel, wind and solar sources and capture real-world fluctuations experienced in the power distribution system. In implementing Algorithm \ref{alg:algorithm}, the step size used is the optimal step size $\eta = 1/N$. In Figure \ref{fig:social}, we illustrate the performance of the OD3 algorithm on the social welfare of the system in relation to the optimal social welfare over time, given fluctuating system dynamics. As can be observed in the plot, the social welfare computed online tracks the optimal social welfare. 
The plot in Figure \ref{fig:power} shows the aggregate online power allocation computed by the OD3 algorithm relative to the optimal aggregate power allocation and the power supply capacity. Notice that due to the equality constraint in \eqref{eq:system-problem1}, the optimal aggregate power is, in fact, the supplier capacity at each iteration. Furthermore, the online aggregate power closely follows the aggregate optimal allocation.

 Figure \ref{fig:price-error}, illustrates the price $\mbf{p}(t)$ generated by the OD3 algorithm and the optimal price $\mbf{p}^*(t)$ at each iteration for the multi-user, single supplier example. The results shown here indicate that the real-time, prices, computed online are very close to their respective optimal values at each time-step, despite the system fluctuations. 

%
%
\section{Conclusions}
\label{sec:conclude}
In this paper, we considered the problem of allocating electric power to users in a power distribution system, where the utility functions of the users and supply capacities of the suppliers are time-varying using an OD3 algorithm. We assumed the changes in the supplier capacities and user utility functions are changing at the the same time-scale as iterations in the algorithm, and presented a worst-case robustness analysis of the OD3 algorithm. In particular, we investigated and characterized performance guarantees of the OD3 power allocation Algorithm by deriving bounds on the social welfare of the system as a function of fluctuations in the system resulting from dynamic supply capacities as well as changes in the users' utility functions. 
The OD3 algorithm uses a one-way message passing protocol between the users and the suppliers in which the suppliers broadcast a coordinating signal (price) and the users locally compute their power allocation based on the price received. Furthermore, we presented and showed convergence of the OD3 algorithm. As illustrated, fluctuations in optimal allocation and price depend on variations in supply capacity and user utility functions. 

\appendices
 
 \section{Proof of Proposition \ref{prop:strongly-convex-dual}}
 \label{app:proof-prop:strongly-convex-dual}
 \begin{proof}
 Let $L$ be the global Lipschitz continuity parameter for all users $i$ and time $t$. Consider the decomposed dual function between users $D(\mbf{p}) = \sum_{i=1}^N D_i(\mbf{p})$, where 
\begin{equation}\nonumber
D^t_i(\mbf{p}) = U_i^t(\mbf{q}_i(\mbf{p}(t))) - \mbf{p}^T\mbf{q}_i(\mbf{p}(t)) + \frac{1}{N} \mbf{p}(t)^TQ(t),
\end{equation}
with gradient
\begin{align}\nonumber
\nabla D_i(\mbf{p}(t)) &= Q(t)/N - \mbf{q}_i(\mbf{p}(t))\\
&= Q(t)/N - (\nabla U_i^t)^{-1}(\mbf{p}(t))
\end{align} 
where $\mbf{q}_i(\mbf{p})$ is defined in \eqref{eq:local-primal}. To prove the Proposition, we show that $D_i(\mbf{p})$ is $1/L-$convex. From \cite[Theorem 2.1.9]{nesterov2004introductory} $D_i(\mbf{p})$ is  $1/L-$convex if and only if for all $\mbf{p}_1$, $\mbf{p}_2$, 
   \begin{multline}
            \frac{1}{L} \| \mbf{p}_2 - \mbf{p}_1 \|^2 \leq    \left\langle \mbf{p}_2 - \mbf{p}_1, \nabla D_i(\mbf{p}_2)- \nabla D_i (\mbf{p}_1) \right\rangle \\
  ~=\left\langle \mbf{p}_2 - \mbf{p}_1 ,\left({-}(\nabla U_i)^{-1}(\mbf{p}_2)\right) - \left({-}(\nabla U_i)^{-1}(\mbf{p}_1) \right) \right\rangle,  \label{eq:lemma2sdfa}
    \end{multline}
     where the equality comes from taking the gradient of the dual function and using the optimal solutions \eqref{eq:local-primal}. 
The inverse $(\nabla U)^{-1}$ exists since $U$ is strongly concave. Hence, $\nabla U_i$ is bijective and we can choose  vectors $\mbf{q}_1,\mbf{q}_2 \in \mathbb{R}^R$ such that $\nabla U_i(\mbf{q}_1) = \mbf{p}_1$ and $\nabla U_i(\mbf{q}_2)= \mbf{p}_2$. Since $U_i$ is $L-$smooth (from Assumption 1), we have that \cite{nesterov2004introductory}:
   \begin{align}
               {-}  \left\langle \nabla U_i(\mbf{q}_2){-} \nabla U_i(\mbf{q}_1) ,\mbf{q}_2 {-}\mbf{q}_1 \right\rangle 
      { \geq   }\frac{1}{L}   \| \nabla U_i(\mbf{q}_2){-} \nabla U_i(\mbf{q}_1)\|^2,  \notag 
    \end{align}
    or by using $\nabla U(\mbf{q}_1)=\mbf{p}_1$ and $\nabla U (\mbf{q}_2)=\mbf{p}_2$, we get
   \begin{align}
     \hspace{-0.1cm}  \left\langle \mbf{p}_2{-} \mbf{p}_1 ,\left({-}(\nabla U_i)^{-1}(\mbf{p}_2)\right){-} \left({-}(\nabla U_i)^{-1}(\mbf{p}_1) \right) \right\rangle 
      { \geq  } \frac{1}{L}   \| \mbf{p}_2{-} \mbf{p}_1\|^2.  \notag 
    \end{align}
    Hence,~\eqref{eq:lemma2sdfa} holds and we can conclude that $D_i(\mbf{p})$ is $1/L$-convex; therefore, $\sum_{i=1}^N D_i(\cdot)$ is $N/L$ strongly convex.
  \end{proof}
  
%

\section{Proof of Corollary \ref{prop:primal-bound}}
\label{app:prop:primal-bound}
\begin{proof}
  Using the fact that $\mbf{q}^*_i(t) = (\nabla U_i^t)^{-1}(\mbf{p}^*(t))$, $\nabla U_i^t(\cdot)$ is $L_i^t$-Lipschitz continuous, and the bound in Theorem \ref{lem:price-bound}, via the triangle inequality one obtains
  \begin{align*}
     \|\mbf{q}_i^*(t{+}1) {-} \mbf{q}_i^*(t) \| = &  \|  [\nabla U_i^{t+1}]^{-1}(\mbf{p}^*(t{+}1)) - [\nabla U_i^t]^{-1}(\mbf{p}^*(t)) \| \\
        \leq &   \|  [\nabla U_i^{t+1}]^{-1}(\mbf{p}^*(t{+}1)) {-} [\nabla U_i^t]^{-1}(\mbf{p}^*(t{+}1))  \|   \\
                & {+}    \|  [\nabla U_i^t]^{-1}(\mbf{p}^*(t{+}1))  {-} [\nabla U_i^t]^{-1}(\mbf{p}^*(t))  \|  \\
                 \leq & \frac{1}{\sigma} ||\mbf{p}^*(t)-\mbf{p}^*(t{+}1)|| +\frac{\alpha}{\sigma}   \\
               \leq & \frac{ L^2}{\sigma^2} \left(\frac{ \gamma}{ N} + \frac{\alpha}{\sigma}\right)  +\frac{\alpha}{\sigma} \qedhere
  \end{align*}
where we have used that $[\nabla U_i^t]^{-1}$ is $1/\sigma$-Lipschitz continuous, see Lemma~\ref{bijLemma}-b), and Eq.~\eqref{inTVlemma:1} in Lemma~\eqref{Lemma:varying} (in Appendix) together with Theorem~\eqref{lem:price-bound}.  
\end{proof}

\section{Proof of Theorem \ref{thm:main-result2}}
\label{app:thm:main-result2}
\begin{proof} 
For simplicity in notation, let $\mathcal{Q} = \|\mbf{q}_i(t+1) - \mbf{q}_i^*(t+1) \|$. Via the triangle inequality, we can express the LHS of \eqref{eq:price-tracking} as
\begin{align*}
\mathcal{Q} &{=} \| (\nabla U_i^{t{+}1})^{{-}1}(\mbf{p}(t{+}1))-(\nabla U_i^{t{+}1})^{{-}1}(\mbf{p}^*(t{+}1)) \|  \\
& \leq \| (\nabla U_i^{t{+}1})^{{-}1}(\mbf{p}(t{+}1))-(\nabla U_i^{t{+}1})^{{-}1}(\mbf{p}^*(t)) \|  \\
& + \| (\nabla U_i^{t{+}1})^{{-}1}(\mbf{p}^*(t))-(\nabla U_i^{t{+}1})^{{-}1}(\mbf{p}^*(t{+}1)) \|  \\
& \leq \frac{1}{\sigma} \|\mbf{p}(t+1) - \mbf{p}^*(t)\| + \frac{1}{\sigma}\| \mbf{p}^*(t) - \mbf{p}^*(t+1) \| \\
&\leq \frac{1}{\sigma}c \|\mbf{p}(t) - \mbf{p}^*(t)   \| + \frac{1}{\sigma} \frac{L^2}{\sigma}\left( \frac{\gamma}{N} + \frac{\alpha}{\sigma} \right) \\
& \leq \frac{c^t}{\sigma} \|\mbf{p}(0) - \mbf{p}^*(0)   \| +  \frac{L^2}{\sigma^2}\left( \frac{\gamma}{N} + \frac{\alpha}{\sigma} \right),
\end{align*}
where the last inequality above comes from Lemma \ref{thm:convergence}.
\end{proof}

\section{Additional Lemmas and Associated Proofs }
\label{app:additional-lemmas}
\begin{lemma}\label{bijLemma}  \label{inTVlemma:1} 
    Suppose that $U_i$ is $\sigma$-concave has $L$-Lipschitz gradient, for all $i=1,\cdots,N$.
    Then the following holds:
    \begin{enumerate}[a)]
       \item $\nabla U_i$ is bijective on $\R^R$ for all $i{=}1,{\cdots}, N$, i.e., $\nabla U_i^{-1}$ exists.
       \item $\nabla U_i^{-1}$ is $1/\sigma$-Lipschitz continuous and  monotone decreasing with parameter $\sigma/L^2$, i.e., for all $\mbf{x}_1,\mbf{x}_2\in \R^R$ 
         $$- \langle \nabla U_i^{-1}(\mbf{x}_1)- \nabla U_i^{-1}(\mbf{x}_2),\mbf{x}_1-\mbf{x}_2 \rangle \geq (\sigma/L^2) ||\mbf{x}_1-\mbf{x}_2||^2.  $$
         \item $\Gamma:=\sum_{i=1}^N [\nabla U_i]^{-1}$ is bijective, $N/\sigma$-Lipschitz continuous and $N\sigma/L^2$ monotone decreasing.
         \item $\Gamma^{-1}$ is $L^2/(N\sigma)$-Lipschitz continuous.
    \end{enumerate}
\end{lemma}
\begin{proof}
We omit the proofs and refer readers to classic texts in analysis such as  \cite{rudin1987real}. 
\end{proof}

\begin{lemma}\label{Lemma:varying}
  Suppose $U_i^t$ are $\sigma$-strongly concave and the gradients are $L$-Lipschitz continuous for all $i$ and $t$
  and Assumption \ref{assum:utility-bound} holds.
  Then for all $\mbf{p},Q\in \R^N$ the following inequalities hold:   
  \begin{align} 
     \Big\| [\nabla U_i^t]^{-1}(\mbf{p}) - [\nabla U_i^{t+1}]^{-1}(\mbf{p}) \Big\| \leq  \frac{\alpha}{\sigma}, \\
     \Big\| \Gamma_t^{-1}(Q) -  \Gamma_{t+1}^{-1}(Q) \Big\| \leq \frac{\alpha L^2}{\sigma^2}, \label{inTVlemma:1} 
  \end{align}
  where $\Gamma(\mbf{p})= \sum_{i=1}^N [\nabla U_i^t]^{-1} (\mbf{p})$.
  
  \end{lemma}
  \begin{proof}
  
      We start by showing that for $\mbf{p}\in \R^N$ it holds that
       \begin{align}    \label{inLemmatimevar:Uin} 
                     \Big\| [\nabla U_i^t]^{-1}(\mbf{p}) - [\nabla U_i^{t+1}]^{-1}(\mbf{p}) \Big\| \leq  \frac{\alpha}{\sigma}. 
       \end{align}   
      By Lemma~\ref{bijLemma}-a) there exist $\mbf{x}_1,\mbf{x}_2\in \R^N$ such that 
      $\nabla U_i^t(\mbf{x}_1)=\nabla U_i^{t+1}(\mbf{x}_2)=\mbf{p}$.
      Hence using the triangle inequality we get that
      \begin{align}
          0 &= \|\mbf{p}-\mbf{p}\| \notag  \\
             &=\| \nabla U_i^{t+1}(\mbf{x}_2) - \nabla U_i^t(\mbf{x}_1) \|  \notag \\
             &\geq   \|\nabla U_i^{t+1}(\mbf{x}_2)- \nabla U_i^{t+1}(\mbf{x}_1)|| -\|\nabla U_i^{t+1}(\mbf{x}_1)-\nabla U_i^{t}(\mbf{x}_1) \|  \notag \\
             &\geq  \sigma \| \mbf{x}_2-\mbf{x}_1\| - \alpha. \label{inLemma:varying-eq1cc}
      \end{align}
      By rearranging~\eqref{inLemma:varying-eq1cc} we get that
      \begin{align}      
             \frac{\alpha}{\sigma} \geq&   || \mbf{x}_2-\mbf{x}_1||=
                     \Big\| [\nabla U_i^t]^{-1}(\mbf{p}) - [\nabla U_i^{t+1}]^{-1}(\mbf{p}) \Big\|.
      \end{align}      
    \quad \quad By summing over~\eqref{inLemmatimevar:Uin} and using triangle inequality we get also that 
    \begin{align}
                \| \Gamma_t(\mbf{p}) - \Gamma_{t+1}(\mbf{p}) \| \leq \frac{\alpha N}{\sigma}.
      \end{align}       
       Now take any $Q\in \R^R$. 
       Using that $\Gamma_t$ and $\Gamma_{t+1}$ are bijective, (from Lemma~\ref{bijLemma}-c)), there exists $\mbf{p}_1,\mbf{p}_2\in \R^N$ such that $Q = \Gamma_t(\mbf{p}_1)=\Gamma_{t+1}(\mbf{p}_2)$. 
       Similarly, we obtain
       \begin{align}
          0 &= \|Q-Q\| \notag  \\
             &=\| \Gamma_{t+1}(\mbf{p}_2) - \Gamma_t(\mbf{p}_1)  \|  \notag \\
             &\geq   \|\Gamma_{t+1}(\mbf{p}_2) - \Gamma_{t+1}(\mbf{p}_1) \| - \|\Gamma_{t+1}(\mbf{p}_1)-\Gamma_t(\mbf{p}_1) \|  \notag \\
              &\geq  \frac{N \sigma}{L^2} \| \mbf{p}_2-\mbf{p}_1\| - \frac{\alpha N}{\sigma}. \label{inLemma:varying-eq3c}
      \end{align}      
      Finally, by rearranging~\eqref{inLemma:varying-eq3c} and using that $\Gamma_t$ and $\Gamma_{t+1}$ are bijective, (from Lemma~\ref{bijLemma}-c)), we get  that
       \begin{align}      
             \frac{\alpha L^2}{\sigma^2} &\geq  \| \mbf{p}_2-\mbf{p}_1 \|=
                    \| \Gamma_t^{-1}(Q) - \Gamma_{t+1}^{-1}(Q)  \|.
     \end{align}    
       
  \end{proof}

\bibliographystyle{unsrt}
{\small
{\footnotesize
\bibliography{refs.bib}}
\end{document}